\documentclass[letterpaper]{article} 
\usepackage{aaai2026}  
\usepackage{times}  
\usepackage{helvet}  
\usepackage{courier}  
\usepackage[hyphens]{url}  
\usepackage{graphicx} 
\usepackage{natbib}  
\usepackage{caption} 
\usepackage{algorithm}
\usepackage{algorithmic}
\usepackage{amsmath}
\usepackage{amsfonts}
\usepackage{amssymb}
\usepackage{amsthm}
\usepackage{subfigure}
\usepackage{color}

\newtheorem{definition}{Definition}
\newtheorem{example}{Example}
\newtheorem{theorem}{Theorem}
\newtheorem{proposition}{Proposition}

\usepackage{newfloat}
\usepackage{listings}
\DeclareCaptionStyle{ruled}{labelfont=normalfont,labelsep=colon,strut=off} 
\floatstyle{ruled}
\newfloat{listing}{tb}{lst}{}
\floatname{listing}{Listing}
\title{Fair Diffusion Auctions}
\author {
    Zixin Gu\textsuperscript{\rm 1},
    Yaoxin Ge\textsuperscript{\rm 1},
    Yao Zhang\textsuperscript{\rm 2},
    Dengji Zhao\textsuperscript{\rm 1}
}
\affiliations {
    \textsuperscript{\rm 1}Key Laboratory of Intelligent Perception and Human-Machine Collaboration, ShanghaiTech University\\
    \textsuperscript{\rm 2}Kyushu University\\
    guzx@shanghaitech.edu.cn, geyx@shanghaitech.edu.cn, zhang@agent.inf.kyushu-u.ac.jp, zhaodj@shanghaitech.edu.cn
}

\usepackage{bibentry}

\begin{document}

\maketitle

\begin{abstract}
    Diffusion auction design is a new trend in mechanism design which extends the original incentive compatibility property to include buyers' private connection report. Reporting connections is equivalent to inviting their neighbors to join the auction in practice. 
    Then, the social welfare is collectively accumulated by all participants: reporting high valuations or inviting high-valuation neighbors.
    Hence, we can measure each participant's contribution by the marginal social welfare increase due to her participation.
    
    Therefore, in this paper, we introduce a new property called \textit{Shapley fairness} to capture participants' social welfare contribution and use it as a benchmark to guide our auction design for a fairer utility allocation.
    Not surprisingly, none of the existing diffusion auctions has ever approximated the fairness, because Shapley fairness depends on each buyer's own valuation and this dependence can easily violate incentive compatibility. 
    Thus, we combat this challenge by proposing a new diffusion auction called \textit{Permutation Diffusion Auction} (PDA) for selling $k$ homogeneous items, which is the first diffusion auction satisfying $\frac{1}{k+1}$-Shapley fairness, incentive compatibility and individual rationality.
    Moreover, PDA can be extended to the general combinatorial auction setting where the literature did not discover meaningful diffusion auctions yet.
\end{abstract}



\section{Introduction}

Auction design over social networks, known as diffusion auction design, explored many novel mechanisms since the first diffusion auction proposed by Li \textit{et al.} (2017)~\cite{ijcai2021p605,LI2022103631,zhao2022mechanism}. 
Different from the traditional auction design, diffusion auction considers buyers' private connections to other buyers. Buyers are asked to report their connections, which is equivalent to inviting their neighbors to join the market in practice (it is assumed that not all buyers are in the market initially), which enlarges the market for other purposes such as improving social welfare or the seller's revenue. 
We want to design a mechanism to 
 incentivize buyers to report their connections truthfully, so that all potential buyers can participate in the auction.
However, this is challenging because buyers' strategical action space has one more dimension (their connections) and all the traditional mechanisms are not suitable to prevent manipulations in the new space.

The first diffusion auction, called Information Diffusion Mechanism (IDM)~\cite{DBLP:conf/aaai/LiHZZ17}, was designed for selling a single item in a network. 
It prioritises cut-points in the network from the seller to the winner and gives them sensible utilities to reward them to connect to the winner.
Despite its novelty, IDM also has an inherent limitation: only the cut points on the path to reach the winner may gain a positive utility and the incentive for the other buyers to join the auction is very weak (their utilities are zero). This limitation weakens the applicability of IDM to actually enlarge a market. 
Thus, Zhang \textit{et al.} (2019)~\cite{zhang2020incentivize} considered about this limitation and proposed a variation to reward more buyers. It allows positive rewards for buyers who are not cut-points but are on some simple paths from the seller to the winner. 
However, this limitation was not fully resolved in their solution: (1) some buyers were still not allowed to get positive rewards; (2) the rewards are not related to their contributions.
The limitation justifies the necessity for fairer diffusion auctions to better motivate all buyers to invite their neighbors to join the auction.

To solve the limitation above, a proper measure for rewarding buyers who are not on paths connecting to the winner is required.
We consider a fair measure to satisfy the following requirements:
\begin{enumerate}
    \item The social welfare is distributed to the agents efficiently;
    \item agents with the same contributions have the same utility;
    \item if one agent always contributes more than the other, then she should receive a higher utility than the other;
    \item agents that do not contribute should receive $0$ utility.
\end{enumerate}

In cooperative game theory, it is proved that there exists a unique value, called Shapley value, satisfying these constraints~\cite{shapley1953value,roth1988introduction,young1985monotonic}.
Similarly, we can define \emph{Shapley contribution} 
as the measure here, which is the average marginal social welfare contribution of a player on all possible joining sequences.

Following the Shapley contribution, we define $\epsilon$-\emph{Shapley fairness} ($\epsilon$-SF), which quantifies how close the buyers' utilities are with their Shapley contributions, to evaluate the fairness of a diffusion auction. 
Under this, 
we show that all the existing diffusion auctions are not Shapley fair, which reflects the limitation we discussed above. The main difficulty is that Shapley fairness depends on each buyer's own valuation and the dependence can easily violate truthful reporting. 

Against this background, we design the very first fair diffusion auction, named \emph{Permutation Diffusion Auction} (PDA), for selling $k$ homogeneous items. PDA imitates the calculation process of the Shapley contribution, but eliminates the dependency of a buyer's payment on her valuation. We show that PDA satisfies $\frac{1}{k+1}$-Shapley fairness, incentive compatibility, and individual rationality.
Moreover, PDA can be generalized to combinatorial settings and still maintains its commitment to fairness, achieving $\frac{1}{n}$-Shapley fairness when selling items in a network with $n$ buyers.

\subsection{Our Contributions}
We study one fairness of diffusion auctions and introduce a new property called Shapley fairness to evaluate the fairness.
We then propose a new mechanism to approximate the fairness in all diffusion auction settings.
Specifically,
\begin{enumerate}
    \item We define the Shapley contribution by mimicing the Shapley value in a diffusion auction setting. Then we use Shapley contribution to define a fairness property. 
    \item For auctioning $k$ homogeneous items, we propose a Permutation Diffusion Auction (PDA) to achieve $\frac{1}{k+1}$-Shapley fair. 
    \item For combinatorial settings, we extend PDA to achieve at least $\frac{1}{n}$-Shapley fair with $n$ buyers. 
\end{enumerate}

\subsection{Related Work}

\noindent\textbf{Single-item Diffusion Auctions.} Information Diffusion Mechanism (IDM)~\cite{DBLP:conf/aaai/LiHZZ17}, the first diffusion auction, was designed for selling a single item in a network and it prevented manipulations in the new space by offering priorities to the cut points to the winner in the network. IDM treated these cut points as sellers, which made them gain a profit from reselling the item to incentivize them to report their connections truthfully. 
Follow-up mechanisms were then proposed.
Critical Diffusion Mechanism (CDM)~\cite{DBLP:conf/ijcai/LiHZY19} was developed from IDM for further improving the seller's revenue.
Fair Diffusion Mechanism (FDM)~\cite{zhang2020incentivize} rewarded non-cut-point buyers with redistribution techniques.
Sybil Cluster Mechanism (SCM)~\cite{chen2023sybil} achieved Sybil-proofness; 
Sequential Resale Auction (SRA)~\cite{liu2023distributed} focused on distributed scenarios;
Closest Winner of Myerson’s (CWM)~\cite{zhang2024optimal} was dedicated to maximizing the seller's revenue.
All existing mechanisms chose certain buyers on the paths from the seller to the winner to offer positive utilities. They are not Shapley fair because some contributed participants would receive zero utilities.

\noindent\textbf{Multi-item Diffusion Auctions.}
Diffusion auctions was also extensively studied in the context of selling multiple homogeneous items.
The Generalized Information Diffusion Mechanism (GIDM), the Distance-based Network Auction Mechanism for Multi-unit, and the Unit-demand Buyers (DNA-MU)~\cite{kawasaki2020strategy} were unsuccessful in guaranteeing IC, thereby rendering the achievement of an IC mechanism a significant challenge.
Two IC mechanisms were then proposed: 
the Layer-based Diffusion Mechanism (LDM)~\cite{liu2023diffusion} localized competition within layers based on the buyers' distances from the seller;
the Multi-unit Diffusion Auction Mechanism (MUDAN)~\cite{fang2023multi} iteratively allocated items to winners during graph exploration. 
None of these mechanisms are Shapley fair as they only allow partial buyers to achieve positive utilities.

\noindent\textbf{Shapley Value in Auctions.}
The balanced winner contribution rule (BWC)~\cite{lindsay2018shapley} was proposed for allocating surplus to winners with modified Shapley value in auctions and exchanges, which is proved to be not incentive compatible.
Nested knockout~\cite{graham1990differential} was proposed to frequently distribute collusive gains among members in a bidder coalition in auctions, giving out the expected payments equal the Shapley value in many settings.
A Shapley Value based profit allocation scheme (SPA)~\cite{pan2009fair} was proposed to distribute the profit among the bidding nodes according to their marginal contributions in the spectrum auction. 
All these work sought for a payoff distribution that matched the Shapley value on behalf of fairness in traditional settings, which was not compatible with truthful reporting. In our paper, we consider a new trending model of diffusion auctions, and aims at a good approximation of Shapley contribution with the hard constraint of incentive compatibility.

\section{Preliminaries}

Consider an auction that a seller $s$ sells $k \geq 1$ homogeneous items in a social network with $n$ buyers, notated as $N$.
The social network is modeled as a graph $G=(V,E)$, where $V = N \cup \{s\}$ and $E$ is the edge set, representing the social relationships between agents. We use $e_{ij}$ to represent the edge between agent $i$ and $j$.

Each buyer $i \in N$ has a private type $\theta_i=(v_i,r_i)$, where $v_i$ is her valuation and $r_i$ shows her social relationship.
$v_i(q)$ denotes agent $i$'s valuation of $q \in \{0, 1, \ldots, k\}$ items, which we restrict to be marginal non-increasing, i.e. $v_i(q+1) - v_i(q) \leq v_i(q) - v_i(q-1)$. We assume that the valuation is normalized, which means $v(0)=0$.
$r_i=\{j \mid e_{ij}\in E\}$ denotes buyer $i$'s direct neighbors in the social network. 
Specially, the neighbors of the seller $r_s$, called initial buyers, are the only buyers knowing that $s$ is selling the item initially.
Let $\Theta_i$ be the type space of buyer $i$ and $\Theta=(\Theta_1,...,\Theta_n)$ be the joint type space of all buyers.

Diffusion auction mechanisms ask each buyer to report not only her valuation but also neighbors, so that the seller can invite more potential buyers to join. Let $v'_i$ denote the reported valuation and $r'_i \subseteq r_i$ be the actual neighbor set reported by $i$. We call $\theta'_i=(v'_i,r'_i)$ the reported type of $i$ and let $\boldsymbol{\theta}'_{-i}$ represent the reported type profile of all buyers except for buyer $i$. Let $\boldsymbol{\theta}'=(\theta'_1,\theta'_2,\ldots,\theta'_n)=(\theta'_i,\boldsymbol{\theta}'_{-i})$ represent the reported type profile of all buyers. 

A buyer can join the auction only when she is invited, but to make the model mathematically clean and also make the related definitions more accessible, we assume all buyers in network report to the mechanism (this does not mean that we require this in practice).
Given all buyers' reports $\boldsymbol{\theta}'$, we need to decide who can actually join the auction. We define the \emph{feasible} set $F(\boldsymbol{\theta}')\subseteq V$ as the set of agents connected to $s$ in the reported network graph $G(\boldsymbol{\theta}')=(V,E(\boldsymbol{\theta}'))$, where $E(\boldsymbol{\theta}')=\{e_{ij}\mid j\in r'_i\}$.
The auction can only use the reports from the feasible set to make the decisions.
For an infeasible buyer $i \notin F(\boldsymbol{\theta}')$, she actually cannot participate in the auction in practice. Now the definition of a diffusion auction mechanism is formalized as follows.

\begin{definition}
    A  diffusion auction mechanism $\mathcal{M}=(\pi,p)$ is defined by an allocation and a payment policy $(\pi,p) = \{(\pi_i,p_i)\}_{i\in N}$, where $\pi_i: \Theta \to \{0,1, \ldots, k\}$ and $p_i: \Theta \to \mathbb{R}$ are the allocation and payment functions for $i$. 
    For all reported type profile $\boldsymbol{\theta}' \in \Theta$, it is satisfied that 
    \begin{enumerate}
        \item for infeasible buyers $i \notin F(\boldsymbol{\theta}')$, $\pi_i(\boldsymbol{\theta}')=0$, $p_i(\boldsymbol{\theta}')=0$,         
        \item for feasible buyers $i \in F(\boldsymbol{\theta}')$, $\pi_i(\boldsymbol{\theta}')$ and $p_i(\boldsymbol{\theta}')$ are independent of the reports of the infeasible buyers, and
        \item $\sum_{i\in F(\boldsymbol{\theta}')}\pi_i(\boldsymbol{\theta'}) \leq k$.
    \end{enumerate}
\end{definition}

We denote the set of all allocations satisfying above constraints given $\boldsymbol{\theta'}$ and $k$ as $\Pi(\boldsymbol{\theta'},k)$. In this definition, $\pi_i(\boldsymbol{\theta}')$ represents the quantity of the items allocated to $i$; $p_i(\boldsymbol{\theta}')$ indicates the amount $i$ should pay to the seller $s$ and $p_i(\boldsymbol{\theta}')<0$ means that buyer $i$ can receive $|p_i(\boldsymbol{\theta}')|$ from the seller $s$. 

Given a reported type profile $\boldsymbol{\theta}'$ and a diffusion auction mechanism $\mathcal{M}=(\pi,p)$,
the \textit{utility} of a buyer $i$ with type $\theta_i=(v_i,r_i)$ is determined by her allocation and payment:
$$u_i(\theta_i, \boldsymbol{\theta}', (\pi,p))= v_i(\pi_i(\boldsymbol{\theta}')) - p_i(\boldsymbol{\theta}').$$
Note that for a randomized mechanism, $u_i$ is calculated in expectation. In the paper, we will define a randomized mechanism by a probability distribution on deterministic ones.

The \textit{social welfare} (SW) of an allocation $\pi$ is 
$$\mathsf{SW}(\boldsymbol{\theta}',\pi)=\sum_{i\in N}v'_i(\pi_i).$$
An allocation is called to be \textit{efficient} if it maximizes SW: 
$$\pi^\ast \in {\arg\max}_{\pi \in \Pi(\boldsymbol{\theta'},k)}\mathsf{SW}(\boldsymbol{\theta}',\pi)$$
with the corresponding social welfare 
$$\mathsf{SW}^\ast(\boldsymbol{\theta'}) = \mathsf{SW}(\boldsymbol{\theta}',\pi^\ast) = {\max}_{\pi \in \Pi(\boldsymbol{\theta'},k)}\mathsf{SW}(\boldsymbol{\theta}',\pi).$$ 

We consider two key properties for a diffusion auction:

\begin{definition}
\label{def:IC}
    A diffusion auction mechanism $(\pi,p)$ is incentive compatible (IC) if for all buyer $i\in N$, all type $\theta_i\in \Theta_i$ all reports $\theta'_i\in \Theta_i$ and $\boldsymbol{\theta}'_{-i}\in \Theta_{-i}$, we have
    $$u_i(\theta_i, (\theta_i,\boldsymbol{\theta}'_{-i}), (\pi,p)) \geq u_i(\theta_i, (\theta'_i,\boldsymbol{\theta}'_{-i}), (\pi,p)).$$
\end{definition} 
\begin{definition}
\label{def:IR}
    A diffusion auction mechanism $(\pi,p)$ is individually rational (IR) if for all buyer $i\in N$, all $\theta_i\in \Theta_i$, and all $\boldsymbol{\theta}'_{-i}\in \Theta_{-i}$, we have $u_i(\theta_i, (\theta_i,\boldsymbol{\theta}'_{-i}), (\pi,p))\geq 0.$
\end{definition}

Intuitively, IC guarantees that buyers report their valuation truthfully and invite all their neighbors; IR guarantees buyers who report truthfully to get non-negative utilities.

Beyond the previous studies, we aim to propose a diffusion mechanism where the payment is determined fairly according to the buyers' contribution to the social welfare. To evaluate this, we formalize the definition of fairness property following the Shapley value, which is a classic solution to share value `fairly' among players in cooperative games~\cite{shapley1953value}.
We use each participant's Shapley contribution for generating social welfare in the network as a reference for utility distribution in a diffusion auction.

Denote $B \subseteq V$ as a coalition of agents.
Let $\mathsf{SW}^\ast(\boldsymbol{\theta}', B)$ be the maximum social welfare can be achieved by coalition $B$, which is the social welfare of the efficient allocation only among the feasible buyers in $B$ who are connected to the seller via $B$ only. Specially, if the seller $s$ is not in $B$, then no buyers are feasible within $B$ and the social welfare can be achieved in $B$ is always zero.
\begin{definition}
    For any coalition $B\subseteq V$, the \emph{feasible set} of $B$, denoted as $F(\boldsymbol{\theta'},B)$, is the set of buyers connected to the seller via $B$ only (specifically, $F(\boldsymbol{\theta'},B)=\emptyset$ if $s\notin B$). Correspondingly, the \emph{maximum social welfare}
    within $B$ is 
    $$\mathsf{SW}^*(\boldsymbol{\theta'},B)=\max_{\pi \in \Pi(\boldsymbol{\theta'},k) \text{~and~} \pi_{i \notin F(\boldsymbol{\theta'}, B)} = 0} \mathsf{SW}(\boldsymbol{\theta'},\pi)$$
\end{definition}

Let $O(V)$ be the set of all permutations of $V$, and $o\in O(V)$ denote an order. 
Let $i\prec_o j$ denote that $i$ precedes $j$ in $o$. 
In order $o$, $o_{\prec i} = \{j\mid j\prec_o i\}$ represents the set of agents preceding $i$, and $o_{\preceq i} = o_{\prec i} \cup \{i\}$.
Then, the \textit{marginal social welfare increase} of $i$ in order $o$ is 
$$\mathsf{MC}_i^o(\boldsymbol{\theta}') = \mathsf{SW}^\ast(\boldsymbol{\theta}',o_{\preceq i})-\mathsf{SW}^\ast(\boldsymbol{\theta}',o_{\prec i}).$$ 

Then, the Shapley contribution of an agent is determined as her averaged marginal social welfare increase on all permutations of $V$.
\begin{definition}
    The \textbf{Shapley contribution} of agent $i\in V$ in a diffusion auction is defined as:
    {\small\begin{align*}
        & \phi_i(\boldsymbol{\theta}')
        = \frac{1}{|O(V)|}\sum_{o\in O(V)}\mathsf{MC}_i^o(\boldsymbol{\theta}') = \\
        & \sum_{B\subseteq V\setminus\{i\}} \frac{|B|!(|V\setminus B| -1)!}{|V|!} [\mathsf{SW}^\ast(\boldsymbol{\theta}', B\cup\{i\}) - \mathsf{SW}^\ast(\boldsymbol{\theta}', B)].
    \end{align*}}
\end{definition}

We now define our new property $\epsilon$-Shapley fair ($\epsilon$-SF) by comparing buyers' utilities with their Shapley contributions.
\begin{definition}
\label{def:e-SF}
    A diffusion auction mechanism $\mathcal{M}=(\pi, p)$ is \textbf{$\boldsymbol{\epsilon}$-Shapley fair} ($\epsilon$-SF), $\epsilon\in (0,1]$, if for all \textbf{buyer} $i\in N$, for all type $\theta_i\in \Theta_i$, all other reports $\boldsymbol{\theta}'_{-i} \in \Theta_{-i}$, we have 
    $$\epsilon \cdot \phi_i((\theta_i,\boldsymbol{\theta}'_{-i})) \leq u_i(\theta_i,(\theta_i,\boldsymbol{\theta}'_{-i}),(\pi, p)) \leq \phi_i((\theta_i,\boldsymbol{\theta}'_{-i})).$$
\end{definition}

\textit{$\epsilon$-SF} requires the utilities of all buyers that report truthfully are at least $\epsilon$ times of their Shapley contributions and should not exceed their Shapley contributions. Specially, if $\epsilon = 1$, then a buyer's utility is equal to her Shapley contribution. A mechanism is not SF, if such $\epsilon$ does not exist. 

\begin{figure}[htbp]
    \centering
    \includegraphics[width=0.8\columnwidth]{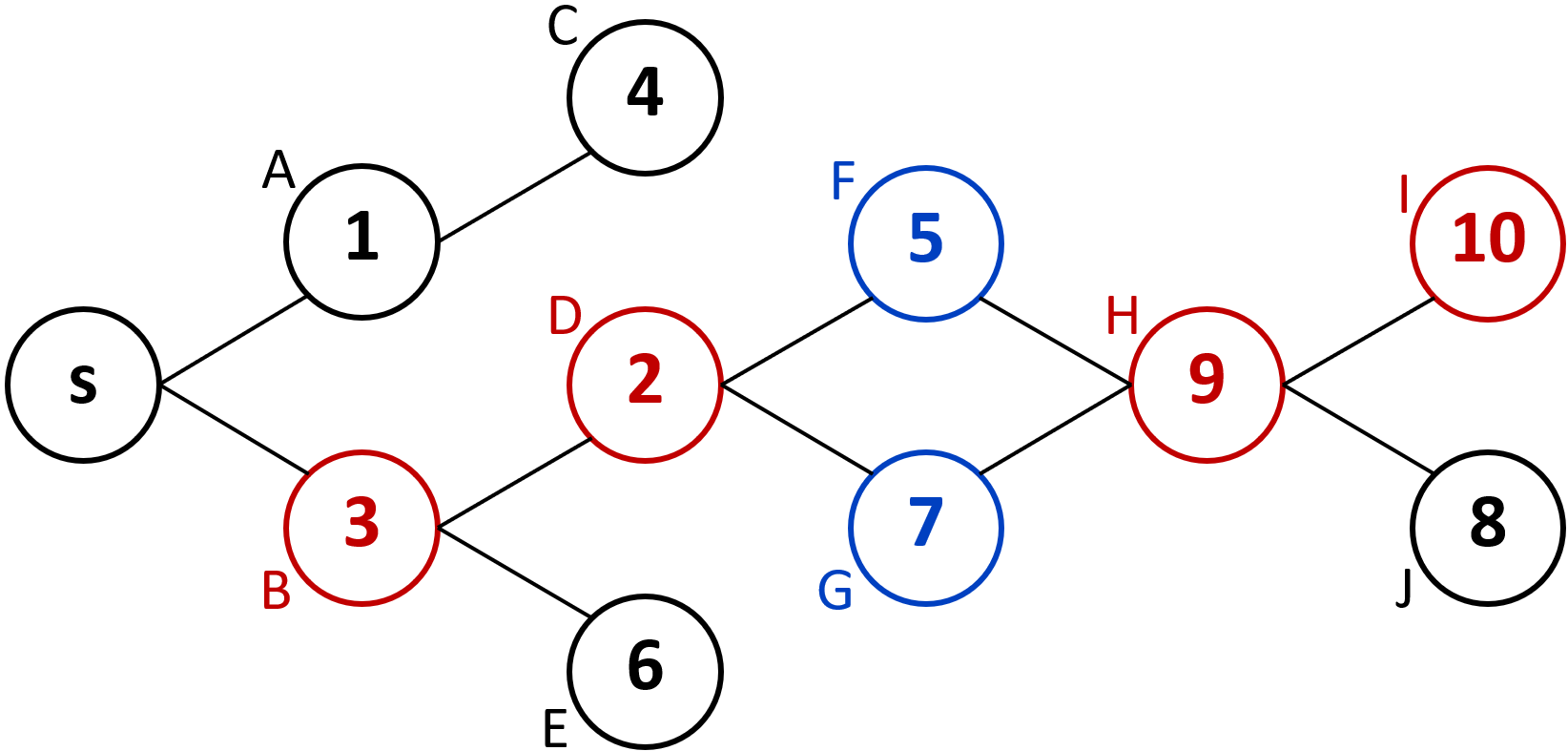}
    \caption{A network example of diffusion auctions. A single item is sold in this network. Each node represents one participant with her valuation in the circle ($s$ is the seller). The edges show the social relationships among them.}
    \label{fig:example-color}
\end{figure}

\noindent\textit{Remark.} \textbf{None of existing diffusion auction mechanisms are Shapley fair.}
It can be showed via a network example shown in Figure~\ref{fig:example-color}.
For single-item mechanisms,
IDM~\cite{DBLP:conf/aaai/LiHZZ17}, CDM~\cite{DBLP:conf/ijcai/LiHZY19}, and SCM~\cite{chen2023sybil} only allow the red buyers (the cut-points from seller $s$ to the buyer who reports the highest value) to have positive utilities; 
CWM~\cite{zhang2024optimal} only allows buyer $H$ (the winner) to have positive utility;
FDM~\cite{zhang2020incentivize} and SRA~\cite{liu2023distributed} allow the red and blue buyers (buyers on some simple paths from seller $s$ to the buyer who reports the highest value) to have positive (expected) utilities.
We can observe that buyers $A, C$, and $E$ always get 0 utilities in all these mechanisms.
However, the Shapley contributions of all buyers except for $J$ in this example are positive. 
Hence, they are not Shapley fair.
For multi-item mechanisms, 
VCG~\cite{vickrey1961counterspeculation,clarke1971multipart,groves1973incentives} for all multi-item settings, and GIDM~\cite{zhao2018selling}, DNA-MU~\cite{kawasaki2020strategy}, LDM~\cite{liu2023diffusion}, MUDAN~\cite{fang2023multi} for homogeneous multi-item setting, are all not Shapley fair.
Among multi-item mechanisms,  VCG~\cite{vickrey1961counterspeculation,clarke1971multipart,groves1973incentives} and GIDM~\cite{zhao2018selling} allow only the cut-points from seller $s$ to winners to have positive utilities; 
LDM~\cite{liu2023diffusion} allows only the first 3 layers of buyers to potentially receive positive utilities; 
MUDAN~\cite{fang2023multi} and DNA-MU~\cite{kawasaki2020strategy} only benefit the winners. 
Therefore, all of them are not Shapley fair.

\section{Permutation Diffusion Auction}

In this section, we propose a diffusion auction mechanism called Permutation Diffusion Auction (PDA), and show that it achieves $\frac{1}{k+1}$-SF, together with IC and IR. 

\subsection{The Definition of PDA}
The main idea of PDA is first randomly select an order from $O(V)$, and then traverse all agents to check their marginal social welfare increase. The process of the traverse is described informally as follows (and formally in Algorithm~\ref{alg:PDA}).

\begin{itemize}
    \item If the current traversed agent $i$ is the seller or she cannot be feasible by only traversed agents, we just skip her; 
    \item If the current traversed buyer $i$ can be feasible by only traversed agents, we calculating an efficient allocation for remaining unsold item (if any) within these agents. If the allocation suggests that $i$ will be allocated with several items, then $i$ will receive these items and pay the loss of social welfare of the coalition before $i$ because of her winning items. If the allocation does not suggest that $i$ could be the winner of some items, we do not perform any allocation in this step, but we will reward buyer $i$ with her marginal social welfare increase.
\end{itemize}

One key point of the above process is that we decide each buyer's allocation and payment when she is traversed instantly, and will never change it in the future. That is to say, we should keep the allocation of sold items when selling the remaining items to a newly traversed buyer, and each buyer's allocation is affected by earlier traversed buyers' allocations. 
More precisely, denote $\pi^{o_{\prec i}}$ as the allocation right before $i$ in order $o$. When $i$ is traversed,
we denote the \emph{efficient} allocation with the irrevocable $\pi^{o_{\prec i}}$ as ${\pi^\ast}(\boldsymbol{\theta}', o_{\preceq i}, \pi^{o_{\prec i}})$ and the corresponding maximized social welfare as $\mathsf{SW}^\ast(\boldsymbol{\theta}', o_{\preceq i}, \pi^{o_{\prec i}})$.
They can be decided by the following optimization problem.

\begin{align*}
    \underset{\pi}{\text{maximize}} \quad&\mathsf{SW}(\boldsymbol{\theta'}, \pi)\\
    s.t. \quad&\pi \in \Pi(\boldsymbol{\theta'}, k) ,\\
                     &\pi_{i \notin F(\boldsymbol{\theta'}, o_{\preceq i})} = 0,\\
                     &\pi_j \geq \pi_j^{o_{\prec i}}, \forall j\in o_{\prec i}.    
\end{align*}

Then, in PDA, we use ${\pi^\ast}(\boldsymbol{\theta}', o_{\preceq i}, \pi^{o_{\prec i}})$ to define the allocation of buyer $i$, i.e., $\pi_i = {\pi^\ast}(\boldsymbol{\theta}', o_{\preceq i}, \pi^{o_{\prec i}})$. Notice that we will not update the allocation $\pi^{o_{\prec i}}$ for those who has been traversed before $i$. The payment for buyer $i$ is the difference of the maximized social welfare for buyers before $i$ in $o$ within the constraints of $\pi^{o_{\prec i}}$ and their welfare in ${\pi^\ast}(\boldsymbol{\theta}', o_{\preceq i}, \pi^{o_{\prec i}})$. For the former one, to be convenient, we use a similar notation $\mathsf{SW}^\ast(\boldsymbol{\theta'}, o_{\prec i}, \pi^{o_{\prec i}})$ to represent the maximized social welfare can be achieved by buyers in $o_{\prec i}$, without decreasing anyone's allocated item in $\pi^{o_{\prec i}}$. For the latter one, it just equals to $\mathsf{SW}^\ast(\boldsymbol{\theta'}, o_{\preceq i}, \pi^{o_{\prec i}}) - v'_i(\pi_i)$.

\begin{algorithm}[htb]
    \caption{Permutation Diffusion Auction (PDA)}
    \label{alg:PDA}
    ~\textbf{Input}: $k, \boldsymbol{\theta}'$
    
    ~\textbf{Output}: $\pi,p$
    \begin{algorithmic}[1] 
        \STATE Initialize $\{\pi_i\}_{i \in N} = 0^N$ and $\{p_i\}_{i \in N} = 0^N$.
        \STATE Select one order $o$  from $O(V)$ uniformly. 
        \FOR{\textbf{each} $i$ in $o$ and $i \neq s$}
            \STATE $\pi_i \leftarrow \pi^\ast_i(\boldsymbol{\theta'}, o_{\preceq i}, \pi^{o_{\prec i}})$
            \STATE $p_i \leftarrow \mathsf{SW}^\ast(\boldsymbol{\theta'}, o_{\prec i}, \pi^{o_{\prec i}}) - \mathsf{SW}^\ast(\boldsymbol{\theta'}, o_{\preceq i}, \pi^{o_{\prec i}}) + v'_i(\pi_i)$
            \STATE \textbf{if} all items are allocated \textbf{then} \textbf{break}
        \ENDFOR
    \end{algorithmic}
\end{algorithm}

We note that PDA could be performed in polynomial time because: 1) we only traverse one randomly selected order so that the number of total steps is at most $n+1$; 2) the efficient allocation in each step can be computed by a polynomial greedy method since the items are homogeneous with diminishing marginal valuations.

\begin{example}
\label{eg: PDA-running}    
    We show a running example of PDA that a seller sells a single item in the network shown as Figure~\ref{fig:example-result}.
    In Figure~\ref{fig:example-result}(a) and Figure~\ref{fig:example-result}(b), assume that agents $C,I,B,s,F,J$, and $D$ have been traversed sequentially. Under this scenario, the item is not allocated among these buyers, and the current feasible set is comprised of $s, B, D$, and $F$, with an associated maximum social welfare of $5$.
    
    Now, if the next traversed agent is buyer $G$ as shown in Figure~\ref{fig:example-result}(a), then she wins the item as her reported valuation $v'_G$ is the highest among all feasible buyers, and her payment equals $5-7+7 = 5$, namely her reported valuation of the item minus her marginal contribution on the maximum social welfare. The remaining buyers that have not been traversed will all get nothing.

    If the next traversed agent is buyer $H$ as shown in Figure~\ref{fig:example-result}(b), then
    she will make $H, I$, and $J$ feasible. The maximum social welfare will rise to $10$ and $H$ will be rewarded by her marginal contribution $10-5 = 5$. After $H$, none of $G, A$, and $E$ will win the item because the maximum social welfare keeps $10$, which is higher than their reported valuations. They will all get 0 payments as their marginal contributions are $0$. We can notice that in this case, the auction ends without allocating the item to any buyers.

    Finally, Since PDA is a randomized mechanism, we consider its expected utility distribution. Figure~\ref{fig:example-result}(c) reveals that the distribution of expected utilities of all buyers in PDA are highly close to their Shapley contributions. 
    
    \begin{figure}[tb]
    \centering
    \begin{subfigure}
        \centering
        \includegraphics[width=.73\columnwidth]{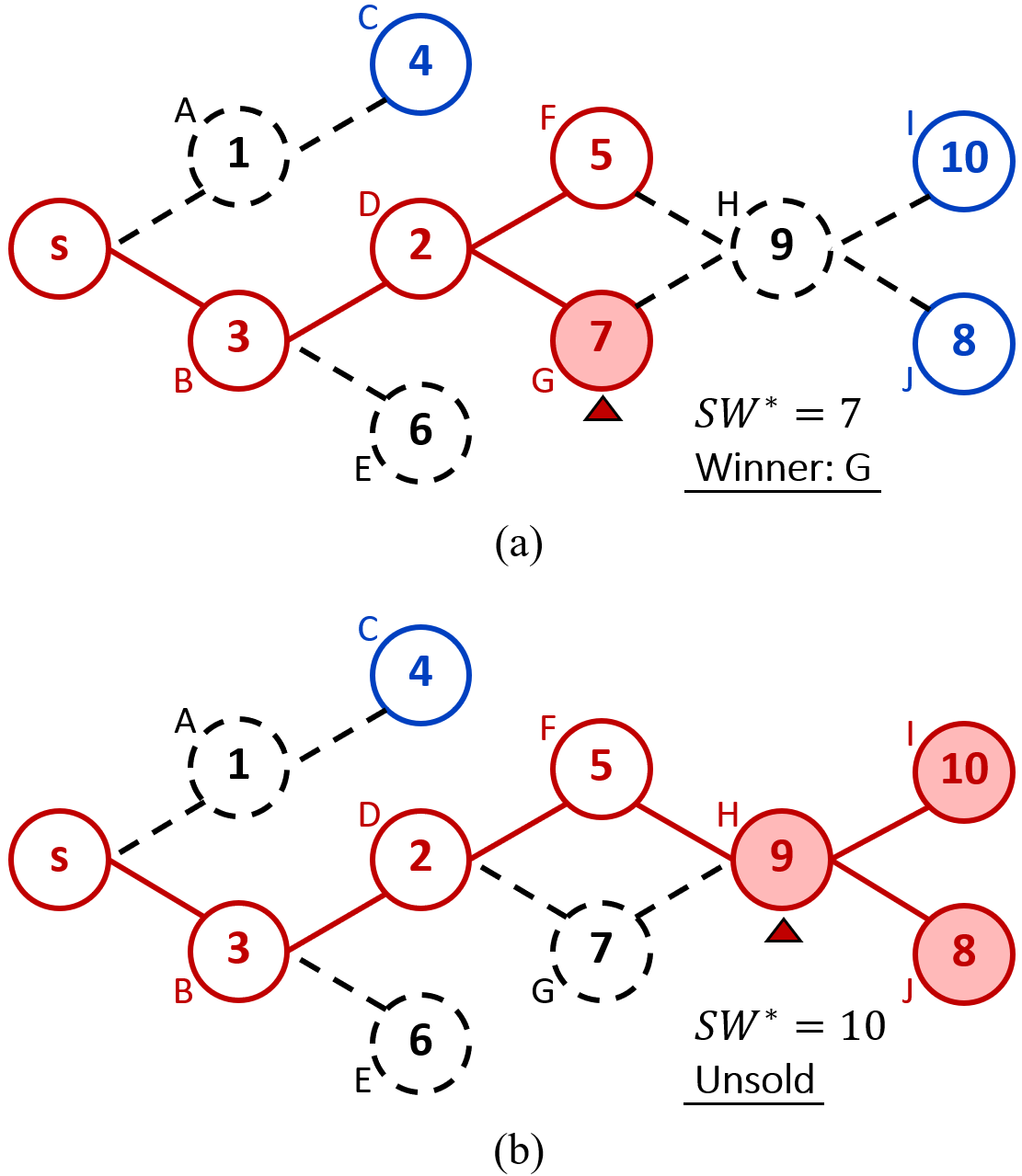}
    \end{subfigure}
    \centering
    \begin{subfigure}
        \centering
        \includegraphics[width=.9\columnwidth]{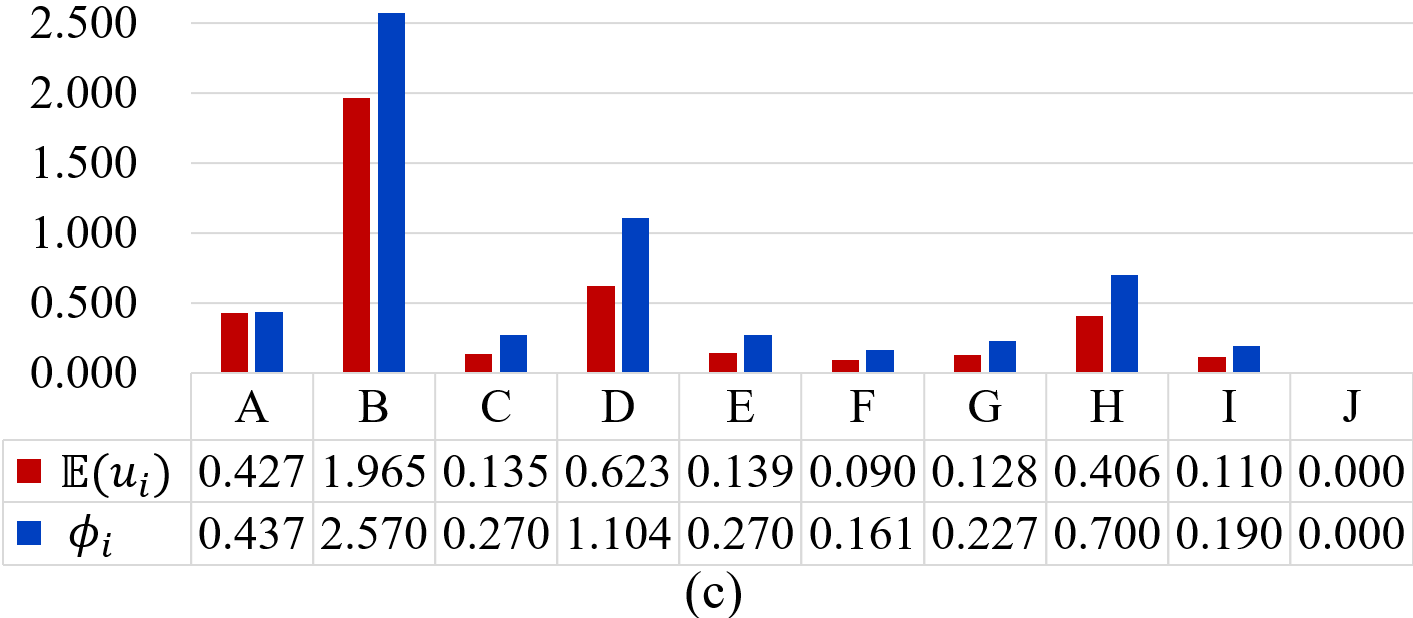}
    \end{subfigure}
    \caption{An example of PDA. (a)\&(b) The dashed nodes represent buyers who have not joined yet, the blue nodes signify buyers who have joined but are not feasible, and the red nodes indicate feasible participants. Specifically, the triangles mark the newly joined buyers, and the filled nodes represent those who have become feasible due to the entry of the new buyers. (c) The expected utility in PDA (red) and the Shapley contribution (blue) of buyers.}
    \label{fig:example-result}
\end{figure}
\end{example}

\subsection{IC and IR Analysis}
Firstly, we show that PDA satisfies the property of IR and IC.
\begin{theorem}
\label{thm:PDA-IRIC}
    PDA is individually rational and incentive compatible.
\end{theorem}

\begin{proof}
    \textbf{IR.} Whichever order selected by PDA, the utility of buyer $i$ with any order $o\in O(V)$ will be
    \begin{align*}
        & \quad u_i^o(\theta_i, \boldsymbol{\theta}') = v_i(\pi_i) - p_i    
        = \ v_i(\pi_i) \\ 
        & \qquad - \mathsf{SW}^\ast(\boldsymbol{\theta'}, o_{\prec i}, \pi^{o_{\prec i}}) + \mathsf{SW}^\ast(\boldsymbol{\theta'}, o_{\preceq i}, \pi^{o_{\prec i}}) - v'_i(\pi_i) \\
        &\overset{(\theta'_i = \theta_i)}{=}  \mathsf{SW}^\ast(\boldsymbol{\theta'}, o_{\preceq i}, \pi^{o_{\prec i}}) - \mathsf{SW}^\ast(\boldsymbol{\theta'}, o_{\prec i}, \pi^{o_{\prec i}}) \geq 0. 
    \end{align*}
    
    Because the maximum social welfare can be achieved is monotone non-decreasing with the enlarging of the set of buyers, we have $\mathsf{SW}^\ast(\boldsymbol{\theta'}, o_{\preceq i}, \pi^{o_{\prec i}}) \geq \mathsf{SW}^\ast(\boldsymbol{\theta'}, o_{\prec i}, \pi^{o_{\prec i}})$; thus the last inequality holds. Then, $u_i^o(\theta_i, (\theta_i, \boldsymbol{\theta}'_{-i}))) \geq 0$ for any sampled $o$.
    Therefore, PDA is IR.

    \noindent\textbf{IC.} Whichever order $o\in O(V)$ selected, for all buyer $i$, all type $\theta_i\in \Theta_i$, all reports $\theta'_i\in\Theta_i$ and $\boldsymbol{\theta}'_{-i}\in\Theta_{-i}$, we denote $\pi_i^o(\cdot)$ and $u_i^o(\cdot)$ as $i$'s allocation and utility function within $o$, separately.

    As for any buyer $i \in N$, $\pi_i^o(\boldsymbol{\theta'}) = \pi^\ast_i(\boldsymbol{\theta'}, o_{\preceq i}, \pi^{o_{\prec i}})$, which maximizes the social welfare when traversing at $i$, 
    \begin{align*}
        \mathsf{SW}^\ast((\theta_i, \boldsymbol{\theta}'_{-i}), o_{\preceq i} & , \pi^{o_{\prec i}}) 
        \geq \mathsf{SW}^\ast((\theta'_i, \boldsymbol{\theta}'_{-i}), o_{\preceq i}, \pi^{o_{\prec i}}) \\
        & - v'_i({\pi}^o_i((\theta'_i, \boldsymbol{\theta}'_{-i}))) + v_i({\pi}^o_i((\theta'_i, \boldsymbol{\theta}'_{-i})),
    \end{align*}
    where the right-hand side 
    means the social welfare achieved in $(\theta_i, \boldsymbol{\theta}'_{-i})$ with the allocation under $(\theta'_i, \boldsymbol{\theta}'_{-i})$.

    On the other hand, since $i \notin o_{\prec i}$, then by definition,
    it has $\mathsf{SW}^\ast((\theta_i, \boldsymbol{\theta}'_{-i}), o_{\prec i}, \pi^{o_{\prec i}}) = \mathsf{SW}^\ast((\theta'_i, \boldsymbol{\theta}'_{-i}), o_{\prec i}, \pi^{o_{\prec i}})$ so that PDA is IC by following inequality:
    \begin{align*}
        u_i^o & (\theta_i, (\theta_i, \boldsymbol{\theta}'_{-i}))\\
        = &\ \mathsf{SW}^\ast((\theta_i, \boldsymbol{\theta}'_{-i}), o_{\preceq i}, \pi^{o_{\prec i}}) - \mathsf{SW}^\ast((\theta_i, \boldsymbol{\theta}'_{-i}), o_{\prec i}, \pi^{o_{\prec i}}) \\
        \geq &\ \mathsf{SW}^\ast((\theta'_i, \boldsymbol{\theta}'_{-i}), o_{\preceq i}, \pi^{o_{\prec i}})  - v'_i({\pi}^o_i((\theta'_i, \boldsymbol{\theta}'_{-i}))) \\
        &\ + v_i({\pi}^o_i((\theta'_i, \boldsymbol{\theta}'_{-i})) - \mathsf{SW}^\ast((\theta'_i, \boldsymbol{\theta}'_{-i}), o_{\prec i}, \pi^{o_{\prec i}}) \\
        = &\ u_i^o(\theta_i, (\theta'_i, \boldsymbol{\theta}'_{-i})).
    \end{align*}
\end{proof}

\subsection{Fairness Analysis}
Now we will analyze the fairness can be achieved by PDA. We show the lower bound of the approximation ratio for Shapley-fairness in the following theorem.
\begin{theorem}
\label{thm:PDA-SF}
    PDA is $\frac{1}{k+1}$-Shapley fair for selling $k$ items. 
\end{theorem}

Before proving the theorem, first notice that the Shapley contribution of a player is the expected marginal social welfare increase without any constraints on previous allocation. It is equivalent to keep a zero allocation $0^N$ as the constraint, so that the Shapley contribution can also be written as:
\begin{align*}
    \phi_i(\boldsymbol{\theta}') 
    = & \sum_{B\subseteq V\setminus\{i\}} \Big( \frac{|B|!(|V|- |B| -1)!}{|V|!} \\
    & \quad \cdot \big(\mathsf{SW}^\ast(\boldsymbol{\theta}', B\cup\{i\}, 0^N) - \mathsf{SW}^\ast(\boldsymbol{\theta}', B, 0^N)\big)\Big).
\end{align*}

Now we prove Theorem~\ref{thm:PDA-SF} through two propositions: we first show that $\epsilon$ is associated to the lower bound of ratio that \emph{all} items are unsold, and then show that this ratio is not less than $\frac{1}{k+1}$.
In the proof, we denote $\mu(\boldsymbol{\theta}')$ as the rate that no items are sold in PDA for reported type profile $\boldsymbol{\theta}'$, and denote $\boldsymbol{\theta}'_B$ as the induced reported type profile that sets all agents' reports who are not in coalition $B$ as $nil$.

\begin{proposition}
    \label{prop:ratio}
    When a buyer reports her true type, the lower bound of the ratio between her expected utility in PDA and her Shapley contribution is no less than the lower bound of rate that all items are unsold, and no greater than one.
\end{proposition}
\begin{proposition}
    \label{prop:unsold}
    The lower bound of the ratio that all items are unsold of PDA is $\frac{1}{k+1}$ for selling $k$ items.
\end{proposition}

\begin{proof}[{Proof of Proposition~\ref{prop:ratio}}]
    For $\forall o \in O(V)$, $\forall i \in N$, given $\boldsymbol{\theta}' = (\theta_i, \boldsymbol{\theta}'_{-i})$,  $i$'s marginal social welfare increase is affected by $\pi^{o_{\prec i}}$, the allocations of buyers preceding $i$. When no items are sold in $o_{\prec i}$ (with probability $\mu(\boldsymbol{\theta}'_{o_{\prec i}})$), i.e., $\pi^{o_{\prec i}}$ is a zero vector, $i$ will get $\mathsf{SW}^\ast(\boldsymbol{\theta}',o_{\preceq i}, 0^N) - \mathsf{SW}^\ast(\boldsymbol{\theta}', o_{\prec i}, 0^N)$ as her utility; if some items have already been sold before traversing to $i$, then $i$'s utility is at least $0$ by IR.
    Hence, for the expected utility of buyer $i$, we have 
    \begin{align*}
        \mathbb{E}_{\mathsf{PDA}}(u_i) 
        \geq & \sum_{B\subseteq V\setminus\{i\}} \Big( \frac{|B|!(|V|- |B| -1)!}{|V|!} \mu(\boldsymbol{\theta}'_{B}) \\
        & \cdot\big(\mathsf{SW}^\ast(\boldsymbol{\theta}', B\cup\{i\}, 0^N) - \mathsf{SW}^\ast(\boldsymbol{\theta}', B, 0^N)\big) \Big).
    \end{align*}
    Compared to the formula of the Shapley contribution, 
    $$\inf_{\boldsymbol{\theta}'_{B} \in \Theta}\mu(\boldsymbol{\theta}'_{B}) \cdot \phi_i ~\leq~ \mathbb{E}_{\mathsf{PDA}}(u_i) ~\leq~ \phi_i.$$
    Therefore, a lower bound of the ratio between an arbitrary buyer's expected utility in PDA and her Shapley contribution is $\inf_{\boldsymbol{\theta}'_{B} \in \Theta}\mu(\boldsymbol{\theta}'_{B})$, which is exactly the lower bound of the rate that all items are unsold.
\end{proof}

\begin{proof}[{Proof of Proposition~\ref{prop:unsold}}]

    If $n \leq k$, then the ratio that all items are unsold is at least $\frac{1}{n+1}\geq \frac{1}{k+1}$ (for the cases where $s$ herself is the last agent in the order). 
    
    If $n > k$, let $M$ be the set of buyers who reports the $k$ highest valuations, which must satisfy $|M|\leq k$. Then, we can first observe that if no item is sold when buyers in $M$ all become feasible in the loops of PDA, then all items are remain unsold after PDA terminates.

    Considering a group of orders where buyers in $M$ all become feasible at the same time when we traverse at $s$, then no items can be sold. By changing the order of the seller $s$ and buyers in $M$ while keeping the order of all other buyers, we can get another group of orders where at least one buyer in $M$ is feasible as long as we traverse at her, and hence at least one item must be sold (to the buyer would be traversed at last among $M$ or some other buyer who has been traversed before her) in these orders. The number of the orders in the latter group is $|M|$ times that in the previous group. Denote the probability of selecting these two groups of orders as $a$; then the conditioned ratio that all items are unsold among these orders is $\frac{1}{|M|+1} \geq \frac{1}{k+1}$.

    We then check the orders that not being included above, notated as $O^\ast$, which totally has a probability of $(1-a)$ to appear. 
    In an arbitrary order $o \in O^\ast$, let $x(o)$ be the one who brings the last buyer among $M$ to be feasible, 
    i.e. $M \subseteq F(\boldsymbol{\theta}',o_{\preceq x(o)})$ but $M \nsubseteq F(\boldsymbol{\theta}', o_{\prec x(o)}) $. We also have $\{s\}\cup M \subseteq o_{\prec x(o)}$.
    As the maximum social welfare reaches the peak as $\mathsf{SW}^*(\boldsymbol{\theta}', o_{\preceq x(o)}, \pi^{o_{\prec x(o)}})$ when traversing to $x(o)$, $x(o)$ and buyers after $x(o)$ cannot win any items. Hence, whether items are sold in $o$ only depends on $o_{\prec x(o)}$. Notice that $M \nsubseteq F(\boldsymbol{\theta}',o_{\prec x(o)})$, so for any order $o'$ where only the order of $o_{\prec x(o)}$ is shuffled in $o$, we have $o'\in O^*$. Namely, the set $\{ o'\mid x(o') = x(o), o'_{\prec x(o')} = o_{\prec x(o)} \}\subseteq O^*$ and we can use $o$ as the representative element to denote the set. Then, $O^*$ can be partitioned by a series of representative orders $o_1$, $o_2$, $\dots$, $o_l$, where $l$ is the number of different representative orders. For each representative order $o$, let $\rho(o)$ be the probability that selects an order belongs to the set represented by $o$ in $O^*$.
    
    Then taking all parts together, the total probability of all items are unsold is,
    $$\mu(\boldsymbol{\theta}') \geq a \cdot \frac{1}{k+1} + (1-a) \cdot \sum_{o\in \{o_1, o_2, \dots, o_l\}} \rho(o) \mu(\boldsymbol{\theta}'_{o_{\prec x(o)}}).$$

    Notice that if for any $o\in \{o_1, o_2, \dots, o_l\}$, we can have    
    $$\mu(\boldsymbol{\theta}'_{o_{\prec x(o)}}) \geq 1/(k+1), \text{ and then}$$ 
    $$\mu(\boldsymbol{\theta}') \geq a/(k+1)+(1-a)/(k+1) = 1/(k+1).$$
    To show that it is true, we can recursively analyze arbitrary $o_{\prec x(o)}$ as above and stop with the following conditions:
    \begin{enumerate} 
        \item The buyers reports the $k$ highest valuation among feasible buyers in $o_{\prec x(o)}$ are all initial buyers, then $a = 1$, and hence $\mu(\boldsymbol{\theta}'_{o_{\prec x(o)}}) \geq 1/(k+1)$;
        \item There are less than or equal to $k$ feasible buyers in $o_{\prec x(o)}$ (i.e. $|F(\boldsymbol{\theta}', o_{\prec x(o)}\setminus\{s\}| \leq k$), then 
        \begin{align*}
            \mu(\boldsymbol{\theta}'_{o_{\prec x(o)}}) & \geq 1/(|F(\boldsymbol{\theta}', o_{\prec x(o)})\setminus\{s\}| + 1) \\
            & = 1/|F(\boldsymbol{\theta}', o_{\prec x(o)})| \geq 1/(k+1).
        \end{align*}
    \end{enumerate}
    
    Therefore, $\mu(\boldsymbol{\theta}'_{o_{\prec x(o)}}) \geq 1/(k+1)$ for arbitrary $o_{\prec x(o)}$, indicating that $\mu(\boldsymbol{\theta}') \geq 1/(k+1)$.    
\end{proof}

\begin{figure}[t]
\centering
\includegraphics[width=0.67\columnwidth]{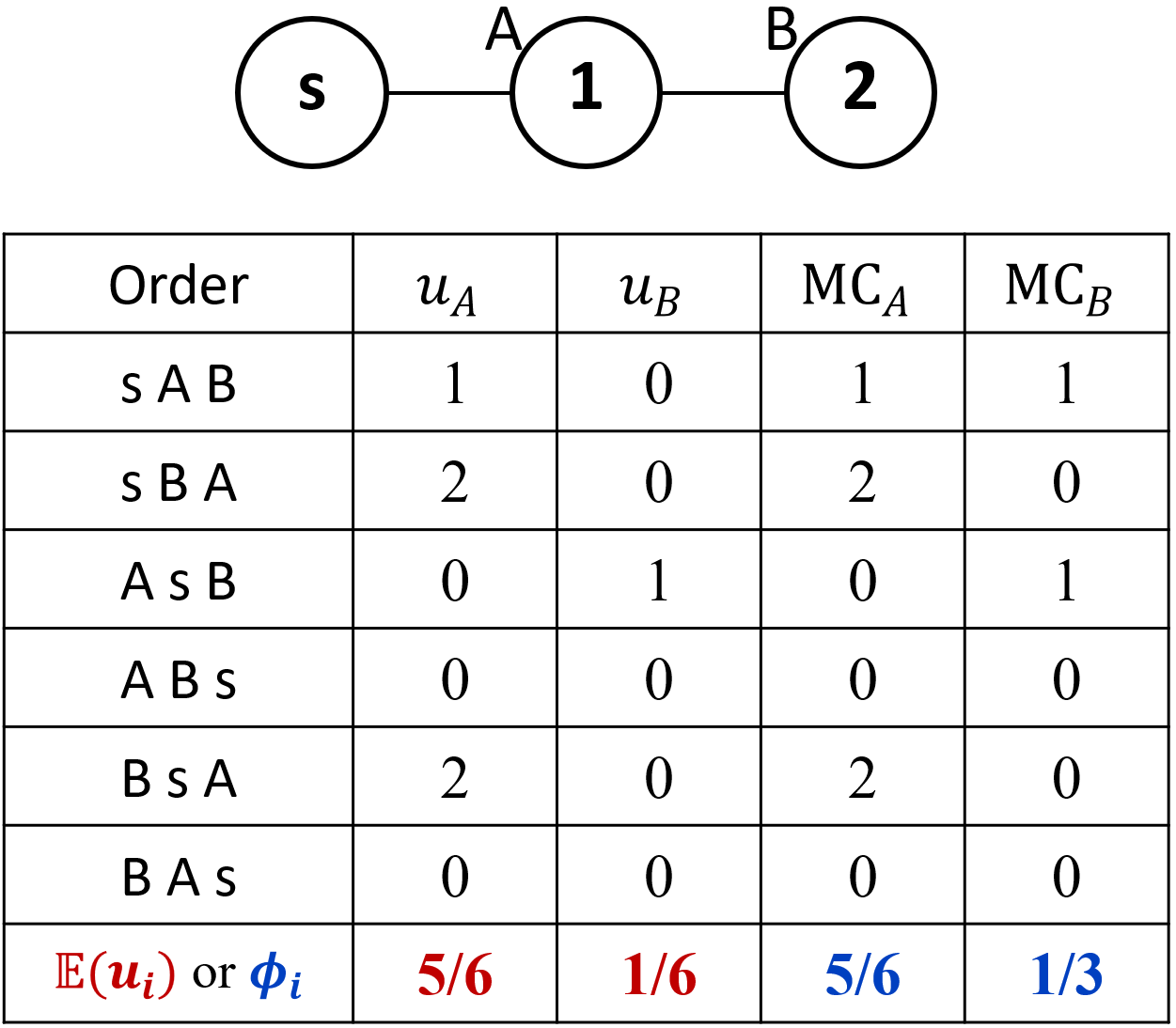}
\caption{An example of PDA that shows $\frac{1}{2}$ is a tight bound when selling a single item. The graph illustrates a single-chain network containing $s, A, B$; the table shows the utility and marginal contribution of $A$ and $B$ under different choices of orders, as well as their expected utilities and Shapley contribution. We can find that the expected utility is half of her Shapley contribution for buyer $B$.}
\label{fig:tight}
\end{figure} 

Finally, Theorem~\ref{thm:PDA-SF} can be proved through Proposition~\ref{prop:ratio} and Proposition~\ref{prop:unsold}.  
We would also like to point out that the bound is consequently tight when $k = 1$ (shown through a simple example in Figure~\ref{fig:tight}), while the tightness of the bound when $k\geq 2$ remains an open problem.

\section{Extension to Combinatorial Auctions}

In this section, we demonstrate that PDA can be extended into a combinatorial setting as Combinatorial PDA (CPDA). 

Consider the setting that a seller $s$ sells $k$ items $\mathcal{K} = \{ 1, 2, \cdots, k\}$ via a social network. 
Indicator vector $q \in \{0, 1\}^k$ denotes a bundle of items, where $q_j$ represents whether item $j$ is in the bundle $q$.
In particular, $q = \boldsymbol{0}$ means the empty bundle, and $q = \boldsymbol{1}$ means the bundle of all items.
Each buyer $i \in N$ has a private valuation function $v_i(q)$ for any bundle $q$, and we assume $v_i(\boldsymbol{0}) = 0$. Now, for a given reported type profile $\boldsymbol{\theta}'$, $\pi_i(\boldsymbol{\theta}')$ is an indicator vector that represents the bundle of items that allocated to buyer $i$. 

Similar to the homogeneous multi-item setting, we need to maintain the allocation of sold items when selling the remaining items. 
We consider the irrevocable $\pi^{o_{\prec i}}$ when maximizing the social welfare in the procedure of CPDA. When traversing to $i$, the efficient allocation $\pi^\ast(\boldsymbol{\theta}', B, \pi^{o_{\prec i}})$ can be solved as following
\begin{align*}
    \underset{\pi}{\text{maximize}} & \ \mathsf{SW}(\boldsymbol{\theta'}, \pi)\\
    s.t. \quad&\pi \in \Pi(\boldsymbol{\theta'}, \mathcal{K}) ,\\
                     &\pi_{i \notin F(\boldsymbol{\theta'}, B)} = \boldsymbol{0},\\
                     &\pi_{jl} \geq \pi_{jl}^{o_{\prec i}}, \forall j\in o_{\prec i}, \forall l\in \mathcal{K}.    
\end{align*}

Let $\mathsf{SW}^\ast(\boldsymbol{\theta}', B, \pi^{o_{\prec i}})$ be the corresponding social welfare.
Then, CPDA follows the same procedure of PDA with a generalized form of $\pi$. 

\begin{algorithm}[htb]
    \caption{Combinatorial PDA (CPDA)}
    \label{alg:CPDA}    
        ~\textbf{Input}: $\mathcal{K}, \boldsymbol{\theta}'$
        
        ~\textbf{Output}: $\pi,p$
    \begin{algorithmic}[1] 
        \STATE Initialize all $\pi_i = \boldsymbol{0}$, $p_i = 0$; uniformly select \small{$o\in O(V)$}. 
        \FOR{$i$ in $o$ and $i \neq s$}
            \STATE $\pi_i \leftarrow \pi_i^\ast(\boldsymbol{\theta}', o_{\preceq i}, \pi^{o_{\prec i}})$
            \STATE {\small $p_i \leftarrow \mathsf{SW}^*(\boldsymbol{\theta}', o_{\prec i}, \pi^{o_{\prec i}}) - \mathsf{SW}^*(\boldsymbol{\theta}', o_{\preceq i}, \pi^{o_{\prec i}}) + v'_i(\pi_i)$}
            \STATE \textbf{if} all items are allocated \textbf{then} \textbf{break}
        \ENDFOR
    \end{algorithmic}
\end{algorithm}

\begin{theorem}
\label{thm:CPDA-IRIC}
    CPDA is IR, IC and $\frac{1}{n}$-Shapley fair.
\end{theorem}

The proof procedure of IR and IC for CPDA is the same as that for PDA, so the full proof is omitted.

For Shapley-fairness of CPDA, we know that it is bounded by the ratio that no items are sold before an arbitrary agent joins according to Proposition~\ref{prop:ratio}.
Given $\boldsymbol{\theta}'$, $\forall o \in O(V)$, $\forall i \in N$, $\mu(\boldsymbol{\theta'}_{o_{\prec i}}) \geq \frac{1}{n}$ because $o_{\prec i} \subseteq V\setminus\{i\}$ contains at most $n$ agents, and no items would be sold when the seller $s$ is the last to participate in $o_{\prec i}$.
The fairness of CPDA is superior to the classic VCG~\cite{vickrey1961counterspeculation,clarke1971multipart,groves1973incentives}, which is not Shapley fair.
Specially, since PDA is a special case of CPDA, then it is also $\frac{1}{n}$-SF, which gives a tighter bound than $\frac{1}{k+1}$ when $n \leq k$. 

\section{Discussion}

This paper designs a diffusion auction that fairly rewards participants to better motivate their involvement. 
To do so, our solution sacrifices some seller revenue, potentially causing deficits when items remain unallocated.
Concretely, with an order $o$, the revenue of the seller $Rev^o(\boldsymbol{\theta'})$ is equal to
{\small\[  \big( \mathsf{SW}(\boldsymbol{\theta'}, \pi^o) - \mathsf{SW}(\boldsymbol{\theta'}, \pi^\ast(\boldsymbol{\theta'}, V, \pi^o)) \big) +  \mathsf{SW}^\ast(\boldsymbol{\theta'}, o_{\preceq s}), \]}
and then the expected revenue of PDA $\mathbb{E}_{\text{PDA}}$ is:
{\small \[ \phi_s(\boldsymbol{\theta'}) - \frac{1}{|O(V)|}\sum_{o \in O(V)} [ \mathsf{SW}(\boldsymbol{\theta'}, \pi^\ast(\boldsymbol{\theta'}, V, \pi^o)) - \mathsf{SW}(\boldsymbol{\theta'}, \pi^o)]. \]}

Intuitively, the expected revenue of PDA equals to the seller's Shapley contribution minus the expected loss of social welfare due to the unsold items. 
The deficit does limit the applicability of PDA if the market's goal is for profit. However, if the goal is getting more people to know the auction, then the deficit could be a kind of advertisement cost.

Unfortunately, the deficit issue appears unavoidable. In a 1-SF mechanism, participants' utilities equal their Shapley contributions, summing to the maximal social welfare. Avoiding a seller deficit requires efficient allocation, but efficiency conflicts with IC and non-deficit~\cite[Theorem 2]{LI2022103631}. Thus, 1-SF is incompatible with these conditions.

Figure~\ref{fig:rev-eff} highlights a trade-off between revenue and allocation efficiency. The revenue upper bound for an $\epsilon$-SF mechanism, $\epsilon \cdot \phi_s + (1-\epsilon) \cdot \mathsf{SW}^\ast(\boldsymbol{\theta'})$, decreases with $\epsilon$, meaning stronger Shapley fairness reduces revenue. An open question remains if non-deficit is a strict constraint:
\begin{center}
    \textit{
        Whether there exists a Shapley fair diffusion auction  \ \ \\
        \ that is IC, IR, and non-deficit? 
    }
\end{center}

\begin{figure}[ht]
    \centering
    \includegraphics[width=.9\columnwidth]{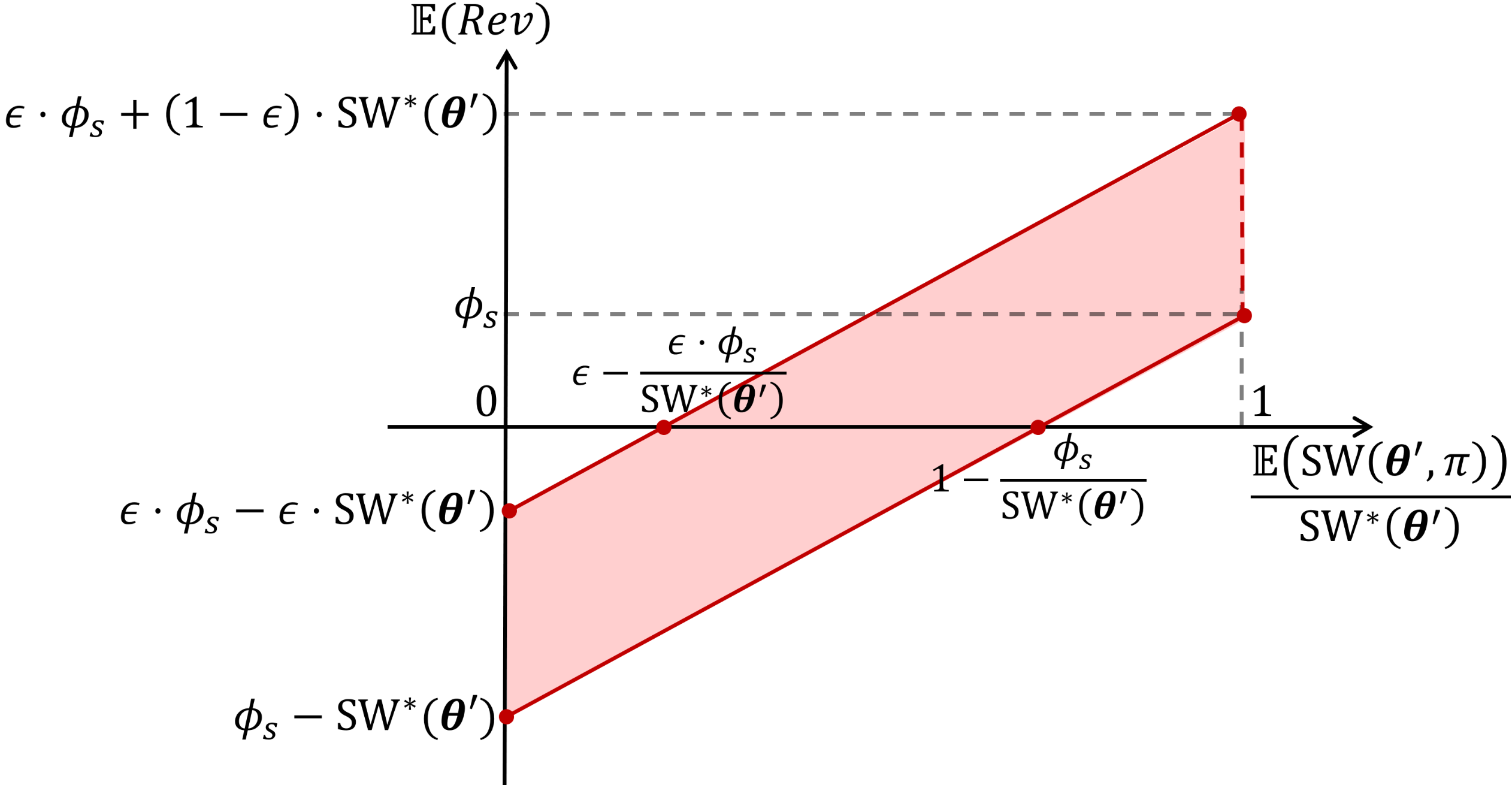}
    \caption{The relationship between expected revenue and the allocation efficiency of Shapley fair diffusion auction mechanisms. The horizontal axis represents the efficiency of allocations and the vertical axis represents the expected revenue. The red parallelogram shows the space of $\epsilon$-Shapley fair mechanisms. Specially, for $1$-Shapley fair mechanisms, the space is an inclined line segment at the bottom.}
    \label{fig:rev-eff}
\end{figure}

\section*{Acknowledgments}
This work was supported by the Science and Technology Commission of Shanghai Municipality (No. 23010503000), the Shanghai Frontiers Science Center of Human-centered Artificial Intelligence (ShangHAI), and JST ERATO (Grant Number JPMJER2301).


\bibliography{aaai2026}
\end{document}